\documentclass[journal]{IEEEtran}
\usepackage{cite}
\usepackage{amsmath,amssymb,amsfonts}
\usepackage{algorithm,algorithmic}
\usepackage{graphicx}
\usepackage{hyperref}
\hypersetup{hidelinks=true}
\usepackage{textcomp}
\def\BibTeX{{\rm B\kern-.05em{\sc i\kern-.025em b}\kern-.08em
    T\kern-.1667em\lower.7ex\hbox{E}\kern-.125emX}}

\usepackage{amsfonts}
\usepackage{cleveref}
\crefname{figure}{fig}{figures}
\Crefname{figure}{Fig.}{Figures}

\newtheorem{definition}{Definition}

\newtheorem{theorem}{Theorem}

\newtheorem{proof}{Proof}
\newtheorem{property}{Property}
\usepackage{subcaption}

\begin{document}
\title{Differentially Private Multimodal Laplacian Dropout (DP-MLD) for EEG Representative Learning}
\author{Xiaowen FU, Bingxin WANG, Xinzhou GUO, Guoqing LIU, Yang XIANG
\thanks{This work is supported by the Project of Hetao Shenzhen-HKUST Innovation Cooperation Zone HZQBKCZYB2020083, and the Project P2070 at HKUST Shenzhen Research Institute, and Shenzhen Science and Technology Program (No. KQTD20180411143338837).  (Corresponding Authors: Guoqing LIU, Yang XIANG).}
\thanks{Xiaowen FU is with Department of Mathematics, The Hong Kong University of Science and Technology, Clear Water Bay, Hong Kong SAR, China (e-mail: xfuak@connect.ust.hk).}
\thanks{Bingxin WANG is with Department of Mathematics, The Hong Kong University of Science and Technology, Clear Water Bay, Hong Kong SAR, China (e-mail: bwangbo@connect.ust.hk).}
\thanks{Xinzhou GUO is with Department of Mathematics, The Hong Kong University of Science and Technology, Clear Water Bay, Hong Kong SAR, China (e-mail: xinzhoug@ust.hk).}
\thanks{Guoqing LIU is with Shenzhen Youjia Innov Tech Co., Ltd., Shenzhen, China (e-mail: guoqing@minieye.cc).}
\thanks{Yang XIANG is with Department of Mathematics, The Hong Kong University of Science and Technology, Clear Water Bay, Hong Kong SAR, China, and Algorithms of Machine Learning and Autonomous Driving Research Lab, HKUST Shenzhen-Hong Kong Collaborative Innovation Research Institute, Futian, Shenzhen, China (e-mail: maxiang@ust.hk).}}

\maketitle

\begin{abstract}
Recently, multimodal electroencephalogram (EEG) learning has shown great promise in disease detection. At the same time, ensuring privacy in clinical studies has become increasingly crucial due to legal and ethical concerns. One widely adopted scheme for privacy protection is differential privacy (DP) because of its clear interpretation and ease of implementation. Although numerous methods have been proposed under DP, it has not been extensively studied for multimodal EEG data due to the complexities of models and signal data considered there. In this paper, we propose a novel Differentially Private Multimodal Laplacian Dropout (DP-MLD) scheme for multimodal EEG learning. Our approach proposes a novel multimodal representative learning model that processes EEG data by language models as text and other modal data by vision transformers as images, incorporating well-designed cross-attention mechanisms to effectively extract and integrate cross-modal features. To achieve DP, we design a novel adaptive feature-level Laplacian dropout scheme, where randomness allocation and performance are dynamically optimized within given privacy budgets. In the experiment on an open-source multimodal dataset of Freezing of Gait (FoG) in Parkinson's Disease (PD), our proposed method demonstrates an approximate 4\% improvement in classification accuracy, and achieves state-of-the-art performance in multimodal EEG learning under DP.
\end{abstract}


\section{Introduction}

Electroencephalography (EEG) is a novel tool for real-time brain activity monitoring \cite{roy2019deep} which bears significant advantages in portability, non-invasiveness, and high temporal resolution \cite{biasiucci2019electroencephalography}. In particular, EEG is often used in clinical settings in diagnosing and monitoring various neurological conditions, including the detection of abnormal sleep patterns \cite{aboalayon2016sleep}, the identification and management of epilepsy \cite{acharya2013automated}, the assessment of attention deficit hyperactivity disorder (ADHD) \cite{arns2013decade}, and moreover, the evaluation of Parkinson’s disease (PD) \cite{maitin2022survey}. Despite the wide applications of EEG data, the analysis of it might suffer from issues induced by low signal-to-noise ratios and non-stationarity. To address these issues, deep learning techniques have shown considerable promise \cite{bigdely2015prep, jas2017autoreject, cole2019cycle, gramfort2013time}. Besides, advances in multimodal data acquisition technologies enable comprehensive collections of various physiological signals, thereby improving model accuracy by integrating richer data inputs \cite{giannakos2019multimodal}, and multimodal deep learning has demonstrated significant potentials in FoG detection with EEG data \cite{skaramagkas2023multi}.

\begin{figure*}[!t]
\centerline{\includegraphics[width=2\columnwidth]{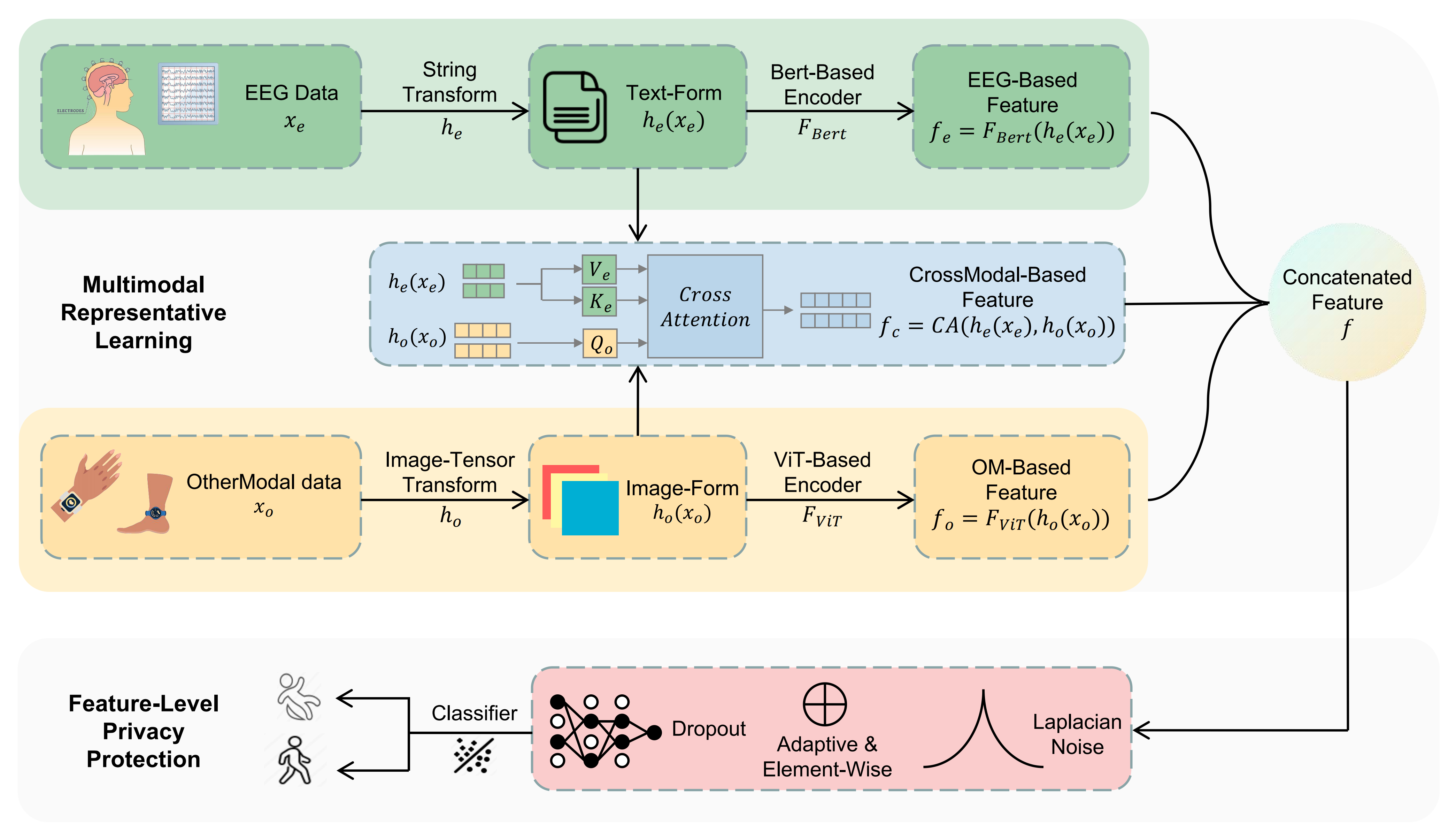}}
\caption{Our proposed DP-MLD for FoG detection: In multimodal representative learning scheme (first three rows), process EEG using BERT as text and other modalities using ViT as images, with cross-attention decoders for feature extraction; In feature-level privacy protection scheme (last row), apply a feature-level Laplacian dropout under DP scheme, optimizing randomness allocation between dropout and Laplacian noise to maximize performance and privacy protection.}
\label{fig_summary}
\end{figure*}


Research on multimodal EEG deep learning 
has laid the groundwork for exploring complex interactions between brain functions and related physiological activities \cite{baig2019survey}. 
Despite these advances, the accuracy of some current methods remains insufficient for reliable clinical application, and the integration of multimodal data is not yet state-of-the-art \cite{baig2019survey}. There is considerable potential for improving techniques for multimodal integration to enhance their effectiveness and robustness. In this paper, we propose a novel approach for multimodal EEG learning, treating EEG data similarly to natural language in language models like BERT and handling other modalities as images in Vision Transformers (ViT). We also introduce cross-attention decoders to extract cross-modal features, and consider privacy protection to ensure the security of sensitive clinical data, which is commonly ignored in previous multimodal EEG studies \cite{baig2019survey}.

Privacy protection is crucial in multimodal EEG studies due to the sensitive nature of the clinical data involved \cite{miotto2018deep}. There exist critical needs for robust privacy-preserving schemes that safeguard private information without compromising performance, and differential privacy (DP) \cite{dwork2006calibrating} stands out as a widely adopted one. DP ensures that the inclusion or exclusion of any single individual data has minimal impact on the outcome of data analysis, thereby limiting the ability of adversaries to infer personal information from aggregated datasets. Various approaches within the DP framework have been proposed, primarily categorized into two streams: DP-SGD \cite{abadi2016deep} and feature-level perturbation \cite{lyu2020differentially}. Adding noise to gradients through the established DP-SGD method \cite{bassily2014private, abadi2016deep} is less feasible for large-scale models due to the increased computational and memory overhead for the large number of parameters, as well as the risk of destabilizing the training process with excessive noise. In contrast, feature-level perturbations \cite{yin2024primonitor, cai2022multimodal, liu2020multimodal}, which add noise directly to features extracted via deep learning models, offer a more straightforward and explicit approach. 

In the context of feature-level privacy protection with DP, some research considers using standard dropout techniques combined with Laplacian noise \cite{lyu2020differentially, yin2024primonitor}. However, there is a notable research gap in applying feature-level DP specifically to multimodal EEG data due to the complexity and sensitivity of neurological information. In particular, whereas dropout can be advantageous by preventing attackers from inferring feature values and making it difficult to determine if a feature is dropped, existing methods that apply a uniform dropout rate across all features \cite{lyu2020differentially} may not provide optimal obfuscation, especially where attackers test it extensively. To this end, we propose to assign an element-wise dropout rate to each feature to accommodate the scenario that some features may be more critical for performance \cite{hou2019weighted, sanjar2020weight, herrmann2020channel} or contain highly sensitive information \cite{biswas2022privacy}. Further, both dropout and adding Laplacian noise to the feature induce randomness in order to protect privacy, so we  consider a dual scheme with element-wise dropout rate and Laplacian noise scale to optimize the performance under given privacy budget, which is an improvement of traditionally sequential privacy budget calculations of dropout and Laplacian noises addition. In this paper, we introduce a novel approach called Laplacian dropout with DP, which involves adaptive randomness allocation to both element-wise dropout and Laplacian noise addition. These parameters are automatically optimized to maximize performance within a given total privacy budget. During training, we utilize a two-step minimization approach for both model parameters and privacy parameters, incorporating Gumbel-Softmax \cite{jang2016categorical} to optimize the privacy parameters. This method is objective-driven and offers a more intuitive privacy guarantee, setting a basis for future advancements in the field.

We apply our method to freezing of gait (FoG) detection for PD diseases. PD is one of the highest ranked common neurodegenerative conditions, with its prevalence rising with age \cite{ma2014prevalence,national2004parkinson}. Among the various symptomatic challenges it presents, FoG is particularly debilitating, affecting over two-thirds of patients in the later stages of the disease \cite{tan2011freezing}. FoG is characterized by sudden and involuntary cessation of movement, which makes patients feel as if their feet are glued to the ground, and often results in falls and significantly reduce the quality of life of patients \cite{nutt2011freezing, walton2015major}. Therefore, The automatic detection of FoG is crucial to ensure the safety of PD patients, and existing literature has shown that EEG is valuable in predicting FoG \cite{wang2020characterization}. In the experiment, we use a multimodal EEG dataset for FoG classification \cite{zhang2022multimodal}, which includes EEG along with other physiological signals such as electromyography (EMG), electrocardiography (ECG), motion data from accelerometers (ACC) and Skin Conductor (SC).

In summary, we propose Differentially Private Multimodal Laplacian Dropout (DP-MLD) for EEG Representative Learning, and the framework is illustrated in \Cref{fig_summary}. Our main contributions are:

\begin{enumerate}

    \item We introduce a novel multimodal EEG learning algorithm that processes EEG data using techniques akin to text models like BERT and other physiological data using image models such as ViT, additional with a decoder for cross-modal feature extraction, marking the first attempt to apply test-image based multimodal learning for EEG and other modal data.
    \item We propose an adaptive Laplacian dropout DP scheme, which optimizes randomness allocation between dropout and Laplacian noise, and makes dropout rate trainable using Gumbel-Softmax, ensuring strong privacy protection while maximizing performance.
    \item Our experiments on an open-source multimodal dataset for FoG detection in PD achieve 98.8\% accuracy, surpassing existing state-of-the-art methods by about 4\%, even though the latter do not consider privacy protection at all.

\end{enumerate}
The remainder of this paper is organized as follows: \Cref{Literature review} gives some background and related works. \Cref{method} introduces our proposed DP-MLD for EEG representative learning. \Cref{Experiments} shows experiment results on the FoG dataset \cite{zhang2022multimodal}. Finally, \Cref{Conclusions} gives conclusions of our method and discusses future work.

\section{RELATED WORKS}
\label{Literature review}
This section covers the background on multimodal EEG and performance comparison on the FoG detection dataset \cite{zhang2022multimodal}, differential privacy (DP), and general multimodal deep learning techniques under DP, including traditional feature-level perturbation methods that inspired our approach. Gumbel-Softmax (GS) technique is also reviewed, which is used in our training. Additionally, we compare the advantages of our method with existing approaches.

\subsection{Multimodal EEG}
Multimodal EEG is extensively studied across various applications, such as emotion recognition \cite{cai2020feature,zheng2014multimodal, zheng2018emotionmeter,guo2019multimodal,lu2015combining,zhao2021expression}, brain-computer interface (BCI) systems development \cite{duan2015design,leeb2010multimodal,ji2016eeg,chen2022multimodal}, epilepsy research \cite{chang2021multimodal,gotman2006combining,memarian2015multimodal,storti2012multimodal}, Alzheimer's disease diagnosis \cite{polikar2010multimodal,liu2013multiple,jesus2021multimodal,zhang2021alzheimer,colloby2016multimodal}, and seizure detection \cite{vandecasteele2021power,furbass2017automatic,sirpal2019fnirs}.

How to improve prediction accuracy of FoG in Parkinson's Disease with multimodal EEG dataset \cite{zhang2022multimodal} is a widely discussed topic. An ensemble technique is proposed and achieve an accuracy of 88.47\% \cite{goel2023ensemble}. A proxy measurement (PM) model based on a long-short-term-memory (LSTM) network achieves an accuracy of $93.6\% \pm 1.8\% $ \cite{guo2022high}. SVM and KNN with handcrafted features achieve an accuracy of $88\%$ \cite{mesin2022multi}.  Multimodal DNN reaches a detection sensitivity of $0.81$ and specificity of $0.88$ \cite{hou2023flexible}. An ensemble model comprising two neural networks (NNs) achieves optimal accuracies of $92.1\%$ one second prior and $86.2\%$ five seconds prior to the event \cite{bajpai2023multimodal}. 
However, there still exists a strong need to improve the prediction accuracy of FoG in Parkinson's Disease given the severe consequences of falls for the old people. Besides, there is a lack of consideration for privacy protection, which is especially critical when utilizing clinical datasets. This oversight highlights the need for integrating robust privacy measures in future multimodal EEG research.

\subsection{Differential Privacy}
Differential Privacy (DP) is a scheme for privacy protection that if changing an individual data in the data set does not cause much change in the outcome, the adversarial may not be able to speculate the real data set as defined in Definition \ref{cp2: def: dp} below.

\begin{definition}
\label{cp2: def: dp}
    (Differential privacy,  \cite{dwork2006calibrating}) 
    A mechanism $\mathcal{A}$ is said to satisfy $\varepsilon$-differential privacy ($\varepsilon$-DP) if for all pairs $\mathbf{x},\mathbf{x}^{'} \in \mathcal{X}$ which differ in only one entry, and for any outcome $O \subseteq \text{range}(\mathcal{A})$, we have 
\begin{equation}
\left|\ln(\frac{Pr(\mathcal{A}(\mathbf{x})\in O)}{Pr(\mathcal{A}(\mathbf{x}^{'})\in O)})\right| \leq \varepsilon.
\end{equation}
\end{definition}

In Definition \ref{cp2: def: dp}, $\varepsilon$ is a parameter controlling privacy leakage and is known as privacy budget. A smaller $\varepsilon$ indicates better privacy protection at the potential cost of statistical accuracy and efficiency. 



To achieve DP, we often need to add some amount of noise to the data or the estimate, and the amount is often determined by the sensitivity defined in Definition \ref{cp2: def:sensitivity}.

\begin{definition}
\label{cp2: def:sensitivity}
    (Sensitivity, \cite{dwork2006calibrating}) 
    The sensitivity of a function $\Gamma$ is the smallest number $S(\Gamma)$ such that for all $\mathbf{x},\mathbf{x}^{'}\in \mathcal{X}$ which differ in a single entry,
\begin{equation}
    \lVert \Gamma(\mathbf{x})-\Gamma(\mathbf{x}^{'})\rVert_1 \leq S(\Gamma),
    \label{intro: sensi_p}
\end{equation}
\end{definition}
where $\|\cdot\|_1$ stands for $L_1$ norm. 

For a random algorithm $\Gamma$, to achieve $\varepsilon$-DP, we often consider the Laplace mechanism in Definition \ref{cp2: def:laplace}. 

\begin{definition}
(Laplace mechanism, \cite{dwork2006calibrating})
\label{cp2: def:laplace}
For all function $\Gamma$ that maps data sets to  $\mathbb{R}^d$, $\Gamma(\mathbf{x}) + \mathbf{w}$ is $\varepsilon$-DP, where $\mathbf{w}=\{w_k\}_{k=1}^d$ is the added Laplacian noise with entry $w_k \sim \text{\text{Lap}}(S(\Gamma)/\varepsilon)$, and $\text{Lap}$ denotes a zero-mean Laplacian distribution with scale $S(\Gamma)/\varepsilon$. 
\end{definition}

We then introduce a property of Laplacian distribution, which will be used in our process of noise generation:

\begin{property}
\label{cp4: property}
    If $t \sim \text{Lap}(1)$, then $ bt \sim \text{Lap}(b)$.
\end{property}

DP for traditional tasks has been widely studied, including count queries \cite{li2010optimizing}, histograms \cite{xu2013differentially}, and other statistical estimates \cite{asi2022optimal}. Recent research on multimodal deep learning under DP, which is closest to our topic, is introduced in the next section.


\subsection{Multimodal Deep Learning under DP}
\label{Feature-Level Perturbation}
Training in DP \cite{dwork2006our} can be broadly categorized into two methods: adding noise to gradients, and to extracted features \cite{yin2024primonitor}. As for the gradient noise injection method, the most widely used scheme is DP-SGD \cite{song2013stochastic,abadi2016deep}, which injects noise into the gradients before each update to the model weights. Further expansion on DP-SGD includes a stateful DP mechanism (DP-FTRL) \cite{kairouz2021practical}, a matrix factorization scheme based on DP-FTRL (MF-DP-FTRL) \cite{choquette2022multi}, etc. However, implementing DP-SGD is not straightforward when the model is large and complex. 

Research in DP also extends to multimodal learning, an area that is rapidly gaining interest. For instance,  the privatization of tensor fusion in multimodal settings is explored, employing techniques like CP decomposition and OME for low-dimensional tensors, and Fourier transform for high-dimensional tensors \cite{cai2022multimodal}. On the other hand, some research proposes adding privacy safeguards to the Contrastive Language-Image Pre-Training (CLIP) model using DP-SGD \cite{huang2023safeguarding}. Some demonstrate effective DP representation learning through image captioning adapted for large-scale multimodal datasets \cite{sander2024differentially}. Additionally, some other research explores applications of DP in federated learning and classification tasks involving multiple data streams, respectively, with a focus on gradient privacy \cite{benouis2023privacy, alshareef2023differential}. However, an efficient DP scheme for multimodal EEG deep learning is still lacking.


As for the feature-level perturbation method, noises are injected to the extracted feature obtained by networks, which is explicit and easy to implement. To the best of our knowledge, one of the most common scheme for feature-level noise injection is from a framework with a word dropout which masks words and the DP
noise injection for NLP learning \cite{lyu2020differentially}. The brief idea of dropout and Laplacian noise injection under DP is detailed as below. 


Let $\mathbf{f}$ denote the extracted representation feature, and $\tilde{\mathbf{f}}$ is obtained via $\mathbf{f}$ after dropout with rate $\mu$. Then Laplacian noise \cite{dwork2006calibrating} is injected to ensure privacy protection:
\begin{equation}
\mathcal{M}(\tilde{\mathbf{f}}) = \tilde{\mathbf{f}} + \mathbf{r},
\end{equation}
where $\mathbf{r} \sim \text{Lap}(b)$, and $b = \frac{S}{\varepsilon^{'}}$. Here, $S$ is the sensitivity of the feature vector $\tilde{\mathbf{f}}$, and $\varepsilon^{'}$ is the temporary privacy budget for $\tilde{\mathbf{f}}$. Given the complexity of accurately estimating true sensitivity,  feature normalization \cite{shokri2015privacy} is applied to $\tilde{\mathbf{f}}$ and then $S = 1$. Under this scheme, $\mathcal{M}(\tilde{\mathbf{f}})$ is $\varepsilon$-DP with final privacy budget $\varepsilon = \ln \left[(1-\mu)\exp (\varepsilon^{'})+\mu \right]$.

The extension of this feature-level perturbation to recurrent variational autoencoder \cite{wang2023differentially}, encoder \cite{maheshwari2022fair}, and random clustered representations perturbations (TextObfuscator)  \cite{zhou2023textobfuscator} are also explored. As for the extension to multimodal learning, a method with application to emotion detection is proposed with a hyperparameter tuning method and Generalized Random Response technique \cite{yin2024primonitor}. However, a naive idea of hyperparameter allocation with Bayesian optimization across modals is used \cite{yin2024primonitor}, and the fixed dropout rate and scale of the privacy parameter may fail to protect privacy if the adversarials attack the algorithm many times \cite{lyu2020differentially}. Thus we argue against uniform dropout rates and Laplacian noise scales for all features.





Current research results show that some features may be more critical for performance \cite{hou2019weighted, sanjar2020weight, herrmann2020channel} or contain highly sensitive information \cite{biswas2022privacy}. Additionally, we recognize the complementary nature of adding Laplacian noise and dropout, both of which introduce randomness for privacy protection. Thus, we advocate for assigning an element-wise dropout rate to each feature and dynamically determining the Laplacian noise scale within the constraints of the total privacy budget. This approach allows for optimal allocation of dropout rates and Laplacian noise, leading to enhanced model performance while ensuring robust privacy protection.

\subsection{Gumbel-Softmax}
\label{sec: GS}
We give the background of Gumbel-Softmax (GS), which is used in our training for dropout parameters. GS \cite{jang2016categorical} is proposed to approximate categorical variables with a differentiable approximation that facilitates gradient-based optimization. It can be derived from the Gumbel-Max technique \cite{gumbel1954statistical}, which involves sampling from a Gumbel distribution to approximate categorical samples. To be specific, suppose $z$ is a categorical variable with class probabilities probabilities \(\pi = (\pi_1, \pi_2, \ldots, \pi_K)\), and categorical samples are encoded as k-dimensional one-hot vectors. The Gumbel-Max technique proposes a perspective of generating samples from the categorical distribution with class probabilities $\pi_i$ with Gumbel variables $g_i \sim \operatorname{Gumbel}(0,1)$ i.i.d, where the Gumbel distribution is defined as \(G \sim -\log(-\log(U))\), with \(U\) being uniform \([0,1]\):

\begin{equation}
    z=\text {onehot}\left(\underset{i}{\arg \max } \left[g_i + \log (\pi_{i})\right]\right).
\end{equation}

Then the Gumbel-Softmax \cite{jang2016categorical} uses an approximation of $softmax$ for $argmax$, which is given by:

\begin{equation}
    y_i = \frac{\exp\left((\log \pi_i + g_i) / \tau\right)}{\sum_{j=1}^K \exp\left((\log \pi_j + g_j) / \tau\right)}, i = 1, \cdots, K,
\end{equation}
where \(\tau\) is a temperature parameter that controls the smoothness of the distribution. As \(\tau\) approaches zero, the Gumbel-Softmax distribution approaches a one-hot encoded vector, closely approximating the categorical distribution. In practice, the Gumbel-Softmax distribution allows gradients to be backpropagated through the discrete choices, enabling the training of models. 

\section{DP-MLD: Multimodal Laplacian Dropout with Differential Privacy}
\label{method}
In this section, we introduce our proposed DP-MLD for EEG representative learning, which includes a novel multimodal EEG learning algorithm and an adaptive feature-level Laplacian dropout scheme under DP.

\subsection{Multimodal Representative Learning}

We present the framework of our proposed method for handling multimodal data, which we categorize into two main modalities: EEG data $\mathbf{x}_e$ and Other Modal (OM) data $\mathbf{x}_o$. In the example application to the FoG classification dataset, EEG data $\mathbf{x}_e$ is collected via devices measuring electrical activity, while OM data $\mathbf{x}_o$ includes ACC and SC data collected by sensors.

We treat EEG data as textual information, considering the inherent sequential nature of both the EEG data and text, and apply Bert \cite{devlin2018bert} to extract semantic features. This approach is supported by findings suggesting that EEG signals can convey textual semantic meaning, which can be effectively extracted using large language models \cite{wang2023large}. Whereas previous work has also explored treating EEG as images and applying models designed for image feature extraction \cite{wang2023large}, our experiments demonstrate that Bert performs better on the dataset used in this research. In contrast to previous exploration \cite{wang2023large}, which primarily focuses on single-modal EEG feature extraction, our study considers the multimodal scenario where OM data is transformed into image tensors through interpolation and reshaping. These tensors are then input into an image-based model such as ViT \cite{dosovitskiy2020image}. We also design the cross-modal feature extraction scheme, which extends beyond previous research that only considers single-modal exploration with only EEG data \cite{wang2023large}. 

Our approach prioritizes configurations that demonstrate superior performance in our experiments, leading us to choose BERT, ViT, and 3-layer decoders for EEG, OM, and cross-modal feature extraction, respectively. We also explore other options, including treating EEG as an image, treating OM as text, using single-stream cross-attention modules with a 12-layer encoder and concatenated embeddings for cross-modal feature extraction, and ALBEF pretraining, which shows less efficiency.

In the following parts, we give a detailed introduction to some key elements of our multimodal EEG learning scheme, which includes processing EEG data as text for language models and other modal data as images for vision transformers, and incorporating well-designed cross-attention mechanisms to effectively extract and integrate cross-modal features.

\subsubsection{EEG Feature Learning via Bert}

Bert (Bidirectional Encoder Representations from Transformers) \cite{devlin2018bert} is a transformative model in natural language processing, capturing deep bidirectional context features. Building upon recent research on the application of large language models to EEG \cite{wang2023large}, which proposes that EEG signals can be processed as text or images to extract semantic features, we adopt the approach of treating EEG data analogously to textual data. We conceptualize each EEG signal as a sequential string of information. 

To be specific, with the EEG signal input from differential channels $\mathbf{x}_{e}\in \mathbb{R}^{d_e}$, we apply a string transformation $h_e$ to it and transform the time-series data to a string data. Then the semantic meaning of the EEG signal is extracted with a Bert-based encoder: 
\begin{equation}
    \mathbf{f}_e = F_{\text{Bert}}(h_e(\mathbf{x}_{e})).
\end{equation}

\subsubsection{Other Modal (OM) Feature Learning via ViT}

As for the other modal data, including the data from Accelerator (ACC) and Skin Conductor (SC), we denote as “OM (other modal)” data for short. When incorporating OM data into our analysis, we use a series of preprocessing techniques, such as interpolation and reshaping skills. These techniques enable us to effectively transform the diverse data types into a unified image tensor format. By transforming the diverse data modalities into a unified image representation, we are able to harness the power of established image models, such as Vision Transformer (ViT) \cite{dosovitskiy2020image}, which have demonstrated remarkable capabilities in various computer vision tasks. 

To be specific, the preprocessing function $h_o$ is applied to transform the raw OM data into an image tensor suitable for processing by ViT. This might involve reshaping and interpolation methods to format the data into a consistent image shape, typically $H \times W \times C$ where $H$, $W$, and $C$ are the height, width, and number of channels of the image, respectively. Then the image semantic of OM data is extracted with a ViT-based encoder:

\begin{equation}
    \mathbf{f}_o = F_{\text{ViT}}(h_o(\mathbf{x}_o)).
\end{equation}

\subsubsection{Cross-Modal Feature Learning}

In order to capture the intricate dependencies and interactions between the EEG and OM data, we propose a multimodal representative learning framework 
to extract interactive features from EEG and OM data, which we refer to cross-modal features. 
The cross-modal features are extracted via cross-attention modules, specifically a decoder. These cross-attention modules facilitate the integration of the EEG and OM data by enabling cross-modal interactions. 

To be specific, we employ $L$ layers decoders with a cross-attention module to extract the cross-attention feature. The cross attention layer $CA_l$ is designed in \eqref{cp4: cross-attention} with $l=1,\cdots,L$ as:
\begin{equation}
    \text{CA}_l = \text{softmax} \left(\frac{W_Qh_o(\mathbf{x}_o)(W_Kh_e(\mathbf{x}_e))^T}{\sqrt{d_k}}W_Vh_e(\mathbf{x}_e)\right), \label{cp4: cross-attention}
\end{equation}
where in the cross-attention module, $W_Q$ is the projection matrix of query $Q$ with OM, and $W_K$, $W_V$ are the projection matrices of key $K$ and value $V$  with EEG, which follow the notations of the attention mechanism \cite{vaswani2017attention}. Then the cross-modal based feature is extracted with $\text{CA}$ constructed with $L$ layers of $\text{CA}_l$: 

\begin{equation}
    \mathbf{f}_c = \text{CA}(h_e(\mathbf{x}_e),h_o(\mathbf{x}_o)).
\end{equation}


\subsection{Laplacian Dropout DP Scheme}
\subsubsection{Element-Wise Laplacian Dropout}

We observe a high adaptability between adding Laplacian noise and dropout to achieve DP. Intuitively, dropping out some features can prevent attackers from inferring their values, even making it difficult to know whether the feature is dropped or not. Though some previous research explores adding Laplacian noise after dropout and calculating the corresponding privacy guarantee level \cite{lyu2020differentially}, we believe that assigning the same dropout probability to each feature may not ensure effective obfuscation. Indeed, an attacker could potentially infer individual features after repeated access to the algorithm. Therefore, we propose incorporating element-wise randomness into our dropout design, adjusting the scale of Laplacian noise accordingly. This approach ensures that the performance is automatically optimized within the specified privacy budget.

Formally, consider we have a $k$-dim feature $\mathbf{f} = (f_1,\cdots,f_k)$ to be privately protected. Suppose we allocate a unique dropout rate to each of them: $\mathbf{w} = (w_1,\cdots,w_k)$. For each feature $f_i$, it suffices to first generate a mask indicator $m_i$ with a Bernoulli distribution:

\begin{equation}
m_i =  \begin{cases}
0, \text{with probability} \quad w_i\\
1, \text{with probability} \quad 1-w_i.
\label{cp4: equa1}
\end{cases}
\end{equation}
Then the feature after dropout is $\mathbf{f}\odot \mathbf{m}$, where $\mathbf{m}$ is the indicator vector $\mathbf{m} = (m_1,\cdots,m_k)$ and $\odot$ means element-wise multiplication.

We then explore the relationship for privacy budget allocation between the element-wise dropout rate and Laplacian scale under the constraint of the total privacy budget in \Cref{thm: dropout}. 

\begin{theorem}
\label{thm: dropout}
    Given the feature $\mathbf{f}$, suppose $\mathcal{M}(\mathbf{f}) = \mathbf{f} + \mathbf{r}$ is $\varepsilon^{'}$-DP by adding element wise Laplacian noise $\mathbf{r} = (r_1,\cdots,r_k)$ with $r_i\sim$ Lap($\frac{1}{\varepsilon_i^{'}}$) and $\varepsilon_i^{'} = \log\left(\frac{exp(\varepsilon)-w_i}{1-w_i}\right)$ to the feature $\mathbf{f}$, and denote the feature after dropout as $\tilde{\mathbf{f}} = \mathbf{f} \odot \mathbf{m}$, then $\mathcal{M}(\tilde{\mathbf{f}}) = \tilde{\mathbf{f}} + \mathbf{r}$ is $\varepsilon$-DP.
\end{theorem}

\begin{proof}
    Suppose there are two adjacent features $\mathbf{f_1}$ and $\mathbf{f_2}$ that differ only in the $j$-th coordinate. For arbitrary masking vector $\mathbf{m}$, we denote the feature after dropout as $\tilde{\mathbf{f}}_1=\mathbf{f_1} \odot \mathbf{m}, \tilde{\mathbf{f}}_2=\mathbf{f_2} \odot \mathbf{m}$. Considering there are two possible cases, i.e., $m_j = 0$, and $m_j = 1$.

    If $m_j = 0$, then after dropout, $\tilde{f}_{1j} = \tilde{f}_{2j} = 0$. Since $f_{1i} = f_{2i}$ for $i \neq j$, then we have $\tilde{\mathbf{f}}_1 = \tilde{\mathbf{f}}_2$. Naturally, we have 
    \begin{equation}
        Pr\{\mathcal{M}{(\tilde{\mathbf{f}}}_1)  = S \} = Pr\{\mathcal{M}{(\tilde{\mathbf{f}}}_2)  = S\}.
    \end{equation}

    On the other hand, if $m_j = 1$, which means the $j$-th feature is not dropped out, then $\tilde{\mathbf{f}}_1$ and $\tilde{\mathbf{f}}_2$ are still only differ in the $j$-th element. According to the definition of DP, 
    we have 
    \begin{equation}
        Pr\{\mathcal{M}{(\tilde{\mathbf{f}}}_1)  = S\} \leq exp(\varepsilon^{'}_j)Pr\{\mathcal{M}{(\tilde{\mathbf{f}}}_2)  = S\}.
    \end{equation}
    Considering the $j$-th feature has the probability $w_j$ to be dropped, we can combine the two cases according to the corresponding probability:
\begin{align}
Pr\{\mathcal{M}{(\tilde{\mathbf{f}}}_1)  = S\}  = &Pr\{\mathcal{M}{(\tilde{\mathbf{f}}}_1)  = S|m_j = 0\} \\
&+ Pr\{\mathcal{M}{(\tilde{\mathbf{f}}}_1)  = S|m_j = 1\} \nonumber \\
\leq & Pr(m_j=0) Pr\{\mathcal{M}{(\tilde{\mathbf{f}}}_2)  = S\} \\
&+ Pr(m_j=1) Pr\{\mathcal{M}{(\tilde{\mathbf{f}}}_2)  = S\} \nonumber \\
= & \left[w_j + (1-w_j)exp(\varepsilon^{'}_j)\right] Pr\{\mathcal{M}{(\tilde{\mathbf{f}}}_2)  = S\} \nonumber \\ 
\leq & exp(\varepsilon) Pr\{\mathcal{M}{(\tilde{\mathbf{f}}}_2)  = S\}
\end{align}
That is to say that $\mathcal{M}(\tilde{f})$ is $\varepsilon$-DP.
\end{proof}

Our innovative element-wise budget allocation for dropout and Laplacian noise addition achieves an inherently adaptive optimization for optimal performance under a given fixed privacy budget. This approach not only increases the difficulty for adversaries attempting to infer features even after multiple accesses to the model, but also makes the dropout rate and Laplacian noise scale trainable for performance rather than being fixed hyperparameters.  

Within our framework, a smaller minimum dropout rate, $w_i$, results in a lower $\varepsilon_i'$, thus requiring the addition of more substantial noise. This supplementary aspect of the dual scheme of dropout and Laplacian noise addition offers options for prioritizing different features. Indeed, research indicates that certain features may be more crucial for performance \cite{hou2019weighted,sanjar2020weight,herrmann2020channel} or contain highly sensitive information \cite{biswas2022privacy}, which supports the advantage of our element-wise adaptation. Our experimental results, as detailed in \Cref{cp4: randomness_exp}, demonstrate the adaptability of our method to these insights, suggesting a correlation between feature magnitude and the allocation of randomness modes. Specifically, features with larger magnitudes are more likely to benefit from enhanced privacy through Laplacian noise, while features with smaller magnitudes are dropped out more frequently. This observation validates the effectiveness of our method in automatically balancing dropout and Laplacian noise for every feature.

Now we compare our method with the classical feature-level perturbation approach introduced in \Cref{Feature-Level Perturbation} and highlight the advantages of our approach in terms of implementation and privacy protection. In the classical method \cite{lyu2020differentially}, the uniform variance of the Laplacian noise $\varepsilon^{'}$ and the dropout rate $\mu$ are set as hyperparameters. As a result, the privacy guarantee for the entire algorithm is calculated as $\varepsilon = \mu + (1-\mu)\exp(\varepsilon^{'})$. Under this model, hyperparameters may be optimized using external techniques such as Bayesian optimization, which can be computationally expensive. Additionally, the risk of exposing features is relatively high due to the uniform hyperparameters for dropout rate and Laplacian noise scale. In contrast, our approach alters the management of privacy parameters. We ensure that the entire algorithm adheres to a predefined privacy guarantee, denoted as $\varepsilon$. This allows each element-wise feature to dynamically decide between remaining inactive or adding noise to optimize performance within the specified privacy constraints. Our method of allocating the privacy budget between dropout and Laplacian noise offers a novel perspective on privacy parameter training, rather than treating them as fixed hyperparameters. Moreover, the element-wise parameter approach provides better privacy protection compared to the uniform parameters used in the traditional method. 




\subsubsection{Two-step Optimization for training} 
Given that the privacy budget is fixed and denoted as $\varepsilon$, the variance of Laplacian noise can be expressed as a function of the corresponding dropout rate $\mathbf{w}$, formulated in vector form as $\varepsilon(\mathbf{w})$. Then the loss function is described with the general model parameters $\mathbf{p}$, the dropout rate parameter $\mathbf{w}$ and the corresponding Laplacian noise parameter $\varepsilon(\mathbf{w})$:

\begin{equation}
\min_{\mathbf{p},\mathbf{w}}L(f(D;\mathbf{p},\mathbf{w},\varepsilon(\mathbf{w}));y).
\end{equation}

$L(f(D;\mathbf{p},\mathbf{w},\varepsilon(\mathbf{w})))$ represents the loss of the downstream task, where $f(D)$ denotes the predictions made by model $f$ on data $D$, and $y$ represents the true labels. For instance, in the FoG classification task, $L$ could be chosen as the cross-entropy loss. 

To make the scheme more concise, we consider Property \ref{cp4: property} of Laplacian noise. Then the feature after dropout and Laplacian noise perturbation can be written as $f_\mathbf{p}(D)\odot \mathbf{m} + \frac{1}{\varepsilon(\mathbf{w})}\mathbf{t}$ with $\mathbf{t} \sim \text{Lap}(1)$. Based on the loss function and the re-formulation, we propose a two-step optimization for training the model:

Step 1, given privacy parameters $\mathbf{w}$, masking parameter $\mathbf{m}$ is generated with the Bernoulli distribution with \eqref{cp4: equa1} accordingly. Then the goal is to optimize model parameters $\mathbf{p}$ as \eqref{4_equa: step1_gradient} where the process is standard as that in common deel learning:

\begin{equation}
\label{4_equa: step1_gradient}
\min_{\mathbf{p}}\mathbb{E}_{\mathbf{w},\mathbf{t}}L(f_\mathbf{p}(D)\odot \mathbf{m} + \frac{1}{\varepsilon(\mathbf{w})}\mathbf{t}).
\end{equation}

Step 2, given model parameters $\mathbf{p}$, our goal is to optimize $\mathbf{w}$, the dropout rates, as \eqref{4_equa: step2_gradient}: 

\begin{equation}
\label{4_equa: step2_gradient}
\min_{\mathbf{w}}\mathbb{E}_{\mathbf{p},\mathbf{t}}L(f_\mathbf{p}(D)\odot \mathbf{m} + \frac{1}{\varepsilon(\mathbf{w})}\mathbf{t}).
\end{equation}

However, training $\mathbf{w}$ with the masking term $\mathbf{m}$ in the loss function poses challenges, and thus we adopt the technique of GS reviewed in \Cref{sec: GS} to make the scheme trainable. To be specific, we introduce a categorical vector $v_i = (v_{i1}, v_{i2})$ inferring the category of being dropped out or not for the $i$-th feature, then the masking vector is set as $m_i = v_{i2}$. Denote the class probabilities of being dropped out as $\pi_i = (w_i, 1-w_i)$, then following the idea of Gumbel-Max technique \cite{gumbel1954statistical}, the onehot categorical vector can be generated as:





\begin{equation}
    v^{\text{hard}}_i=\text {onehot}\left(\underset{j}{\arg \max}\left[g_j+\log (\pi_{ij})\right]\right), j=1,2,
    \label{cp4: dropout_hard}
\end{equation}
where $g_j$ are i.i.d samples drawn from $\operatorname{Gumbel}(0,1)$, i.e $g_j=-\log \left(-\log \left(u_j\right)\right)$, $u_j \sim U(0,1)$ for $j=1,2$, and “hard” means generating the vector via $argmax$. When testing, the masking vector $m_i = v_{i2}$ can be generated via \eqref{cp4: dropout_hard}. But when training, the challenge of backpropagation for $argmax$ arises. We adopt GS technique \cite{jang2016categorical}, which uses the $softmax$ function as a continuous, differentiable approximation to $argmax$. We generate categorical vector $v^{\text{soft}}_i$ as in \eqref{cp4: dropout_soft}, where the “soft” means generating the vector via $softmax$.

\begin{equation}
v^{\text{soft}}_{ij}=\frac{\exp \left(\left(g_j+\log \left(\pi_{ij}\right)\right) / \tau\right)}{\sum_{k=1}^2 \exp \left(\left(g_k+\log \left(\pi_{ik}\right)\right) / \tau\right)}, j=1,2. \label{cp4: dropout_soft}
\end{equation}


To summarize, with given dropout rate $\mathbf{w}$, we can generate the masking term variable $\mathbf{m}$ from function $m(v(\pi,u|\mathbf{w}))$. where $v(\pi,u|\mathbf{w})$ is a function mapping from class probabilities $\pi = \pi(\mathbf{w})$ and uniform drawn variable $u$ to the categorical vector $\mathbf{v}$ via \eqref{cp4: dropout_soft} in training and via \eqref{cp4: dropout_hard} in testing. With $\mathbf{v} =v(\pi,u|\mathbf{w})$, $m(\mathbf{v})$ is a function to generate masking variable $\mathbf{m} = m(\mathbf{v})$ with $m_i = v_{i2}$. Then the optimization with regards to the dropout parameter $\mathbf{w}$ is:

\begin{equation}
\label{4_equa: step2_gradient_re}
    \min_{\mathbf{w}} \mathbb{E}_{\mathbf{t},u}L\left[f_{\mathbf{p}}(D)\odot m(v(\pi,u|\mathbf{w}))+\frac{1}{\varepsilon(\mathbf{w})\mathbf{t}}\right]
\end{equation}

\subsection{Algorithm Summary}

To summarize, we first design a multimodal learning mechanism summarized in \Cref{cp4: alg_feature} to extract representative features, and then apply the two-step optimization for Laplacian dropout with DP. The whole process of our proposed method is summarized in \Cref{cp4: alg_summary}.

\begin{algorithm}[H]
\caption{Multimodal EEG representative learning $\mathcal{A}_m$}\label{cp4: alg_feature}
\begin{algorithmic}[1]
\STATE \textsc{Input:} EEG data $\mathbf{x}_e$, OM (other modal) data $\mathbf{x}_o$ 
\STATE \textbf{Extract EEG features} via text-based encoder Bert with text-targeted transformation: $\mathbf{f}_e = F_{\text{Bert}}(h_e(\mathbf{x}_{e}))$\;
\STATE \textbf{Extract other modality features} via image-based encoder ViT with image-targeted transformation: $\mathbf{f}_o = F_{\text{ViT}}(h_o(\mathbf{x}_o))$\;
\STATE \textbf{Extract cross-modal features} via cross-attention-based decoder: $\mathbf{f}_c = \text{CA}(h_e(\mathbf{x}_e), h_o(\mathbf{x}_o))$\;
\STATE \textbf{Combine and normalize features}: 

$\mathbf{f} = \text{Normalization}([\mathbf{f}_e, \mathbf{f}_o, \mathbf{f}_c])$
\STATE \textsc{Output:} Multimodal feature $\mathcal{A}_m(\mathbf{x}_e,\mathbf{x}_o) = \mathbf{f}$
\end{algorithmic}
\end{algorithm}

\begin{algorithm}[H]
\caption{DP-MLD for EEG representative learning}
\label{cp4: alg_summary}
\begin{algorithmic}[1]
\STATE \textsc{Input:} EEG data $\mathbf{x}_e$, other modal data $\mathbf{x}_o$, label $\mathbf{y}$.
\STATE \textbf{Extract multimodal feature} via \Cref{cp4: alg_feature} 
$\mathcal{A}_m(\mathbf{x}_e,\mathbf{x}_o)$.\;
\STATE \textbf{Apply Laplacian dropout} to get noised feature 

$\tilde{\mathbf{f}}=\mathcal{A}_m(\mathbf{x}_e,\mathbf{x}_o)\odot g_m(z(v,u|\mathbf{w}))+\frac{1}{\varepsilon(\mathbf{w})\mathbf{t}}$.\;
\STATE \textbf{Two-step optimization} with loss $L( \text{Classifier}(\tilde{\mathbf{f}}),\mathbf{y})$ via \eqref{4_equa: step1_gradient} and \eqref{4_equa: step2_gradient_re} and using soft \textbf{Gumbel-Softmax} applied in training process.
\end{algorithmic}
\end{algorithm}

\section{Experiments}
\label{Experiments}
\subsection{Dataset Description}
We conduct experiments on an open-source multimodal dataset of FoG in PD \cite{zhang2022multimodal}. This dataset encompasses data from multiple sources: electroencephalogram (EEG), skin conductance (SC), and acceleration (ACC), collected during walking tasks using various sensors. Twelve PD patients participated in the experiments, generating a total of 3 hours and 42 minutes of valid data. They first sort out the original data header, and unify the data sampling frequency to 500Hz. Then they segment and label the data. After this preprocess, there are 3003 datasets, each containing EEG data with 30 channels, OM data with 25 dimensions (including 24 dimensions of ACC data and 1 dimension of SC), and corresponding labels indicating the presence or absence of FoG. We split the dataset into 70\% training data and 30\% test data. Additionally, we use accuracy and macro-F1 score as the metrics for comparison. Accuracy measures the proportion of correct predictions out of the total predictions. Macro-F1 calculates the average F1 score across all classes, treating each class equally.

\subsection{Overall Performance}
The performance comparison on the FoG dataset \cite{zhang2022multimodal} is presented in \Cref{cp4: tb_result}, where our proposed DL-MLD demonstrates significant advantages despite the additional considerations for privacy protection. DL-MLD shows an approximate 4\% improvement in accuracy with a common standard privacy budget $\varepsilon=1.0$. This highlights the dual strength of our approach: enhancing performance while ensuring robust privacy protection, thereby establishing a new basis for future advancements in the field.

Besides, considering that while our method achieves privacy protection, the related works we compare with does not consider privacy protection. Therefore, for fairness and generalization, we also propose a “non-private” variant of our method, which preserves the multimodal learning scheme and removes the feature-level DP scheme. Notably, our non-private variant also shows a significant improvement, enhancing performance by approximately 5\%. This result again underscores the effectiveness of our proposed multimodal representative deep learning scheme.

\begin{table}
\begin{center}
\caption{Accuracy comparison of methods.}
\setlength{\tabcolsep}{3pt}
\begin{tabular}{|l|l|l|}
\hline Method & Accuracy (\%) & Macro-F1 (\%) \\
\hline Depth-width CNN \cite{hou2023flexible} (Hou, et al.) & 85.00 & 78.00 \\
ETC \cite{goel2023ensemble} (Goel K, et al.)& 88.47 & 88.30 \\
SVM \cite{mesin2022multi} (Mesin, Luca, et al.)& 88.00 & 86.73 \\
EEGFoGNet \cite{bajpai2023multimodal} (Bajpai, R., at.) & 92.10 & 85.00 \\
pmEEG-ACC \cite{guo2022high} (Guo, Yuzhu, et al.)& $93.6 \pm 1.8$ & $89.4 \pm 7.7$ \\
\textbf{DP-MLD (Non-private)} & $\mathbf{99.30}$ & $\mathbf{99.21}$ \\
\textbf{DP-MLD (Private with $\varepsilon = 1.0$)} & $\mathbf{98.70}$ & $\mathbf{98.23}$ \\
\hline
\end{tabular}
\label{cp4: tb_result}
\end{center}
\end{table}


\subsection{Privacy Budget}
In this section, we evaluate the performance of our proposed DP-MLD under three different privacy budgets, denoted by $\varepsilon = 0.01, 0.1,$ and $1.0$. Each set of experiments is conducted over 50 epochs, and the corresponding accuracy and loss metrics for both training and testing phases are depicted in \Cref{cp4: fig_epoch}.

For the highest privacy budget of $\varepsilon = 1.0$, our model quickly demonstrates strong performance, achieving an accuracy of $0.956$ by the $10$-th epoch, and it continues to improve, reaching a peak accuracy of $0.987$. This indicates robust model performance with minimal privacy constraints.

Conversely, at a lower privacy budget of $\varepsilon = 0.1$, the model exhibits a slower rate of convergence. By the $10$-th epoch, it achieves an accuracy of only $0.760$. Nevertheless, it eventually attains a comparable peak accuracy of $0.956$, suggesting that whereas initial learning is hampered, the model is capable of recovering to deliver high accuracy as training progresses.

At the most stringent privacy budget of $\varepsilon = 0.01$, the impact on model performance is more pronounced. The accuracy at the $10$-th epoch is just $0.678$, and the model only manages to reach the best accuracy of $0.806$ throughout the 50 epochs.

\begin{figure*}[htbp]
\centerline{\includegraphics[width=1.9\columnwidth]{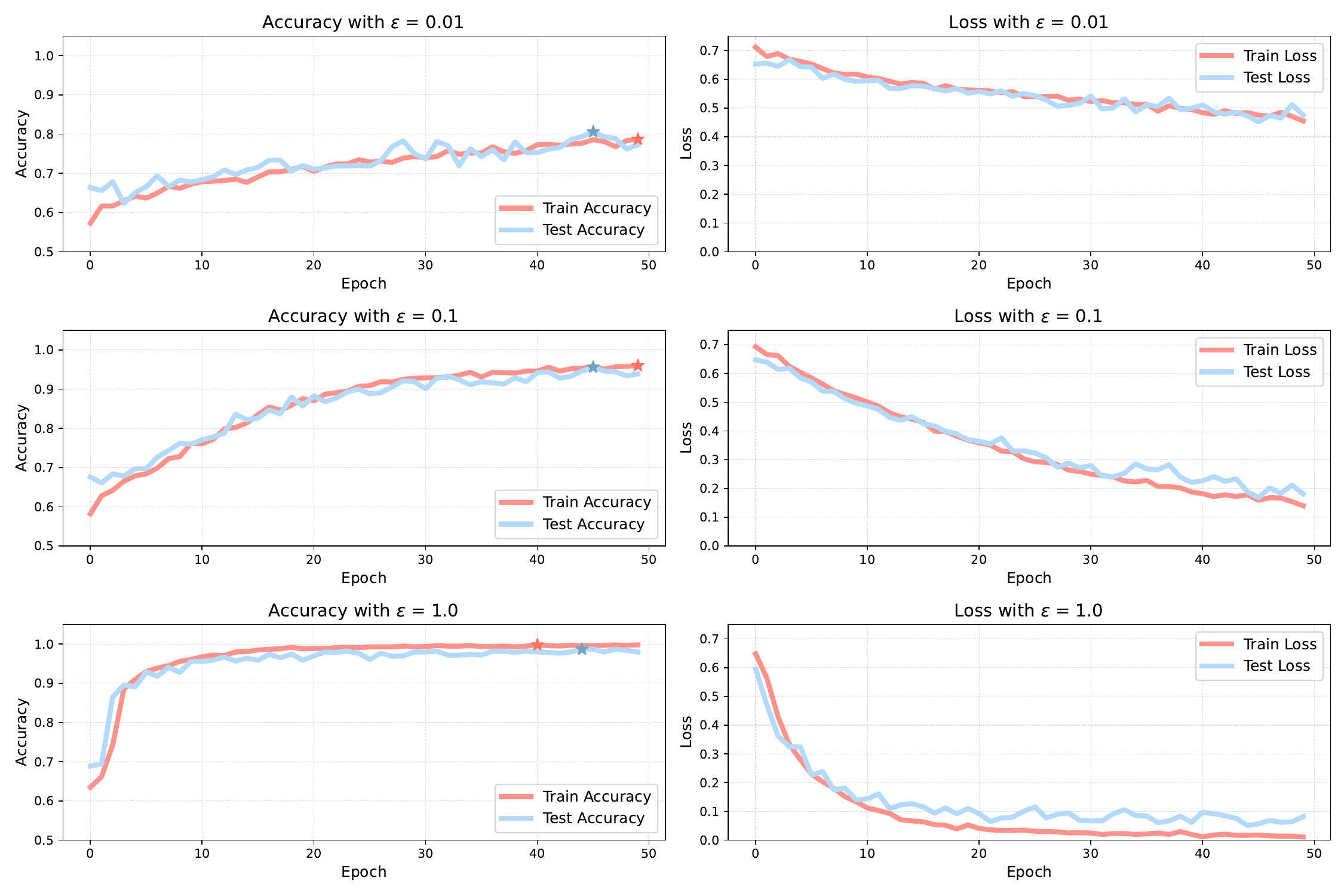}}
\caption{Accuracy and loss in training and testing with different privacy budgets.}
\label{cp4: fig_epoch}
\end{figure*}

These results clearly illustrate the trade-off between performance and privacy protection in machine learning models. As the privacy budget decreases, the ability to learn effectively and to achieve higher accuracy is compromised, highlighting the challenges posed by stringent privacy constraints. A privacy budget of $\varepsilon = 0.1$ represents a balanced choice between performance and privacy protection in our method. A smaller privacy budget indicates more stringent privacy protection. As illustrated in previous work \cite{lyu2020differentially}, a common standard for privacy budget is $\varepsilon = 1.0$. Using our method, excellent results can be achieved at this standard. Moreover, we can even get stronger privacy protection $\varepsilon = 0.1$ with excellent results. This demonstrates that our method outperforms traditional approaches not only in performance, as shown in \Cref{cp4: tb_result}, but also in providing stronger privacy protection with a lower privacy budget, while maintaining excellent performance.

\Cref{cp4: fig_best} demonstrates the best test accuracy achieved for different privacy budgets, highlighting a distinct trend: as the privacy budget increases, so does the test accuracy, albeit up to a certain limit. This trend is particularly noticeable given the restriction on the total number of training epochs, e.g. 10 epochs shown in the figure. The graph effectively illustrates the classic trade-off between performance and privacy. An overall exploration of training and testing accuracy across epochs for different privacy budgets is shown in \Cref{best}.

\begin{figure}[htbp]
\centerline{\includegraphics[width=\columnwidth]{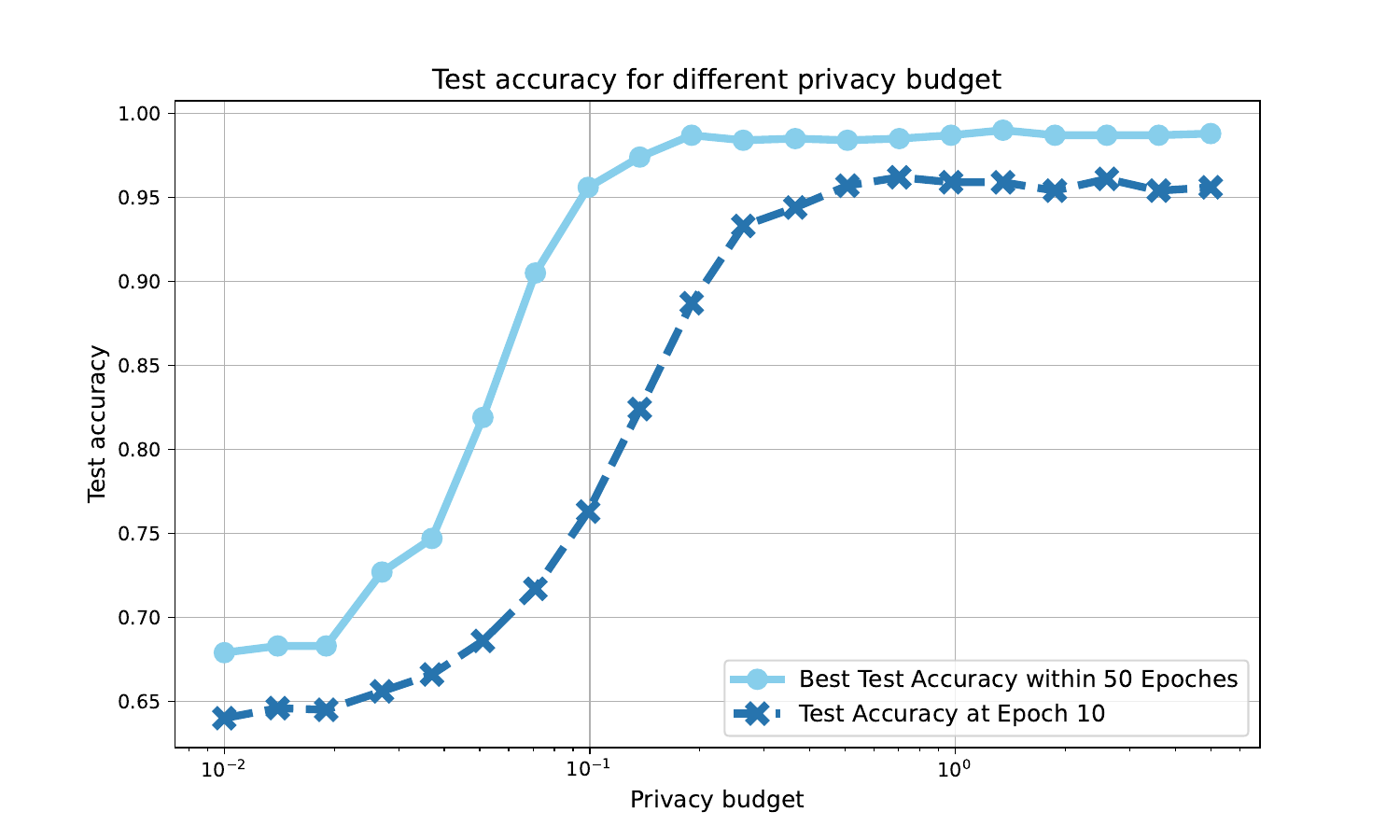}}
\caption{Best accuracy with different privacy budgets in testing.}
\label{cp4: fig_best}
\end{figure}


\begin{figure*}[!t]
    \centering
    \subfloat[Train accuracy]{\includegraphics[width=\columnwidth]{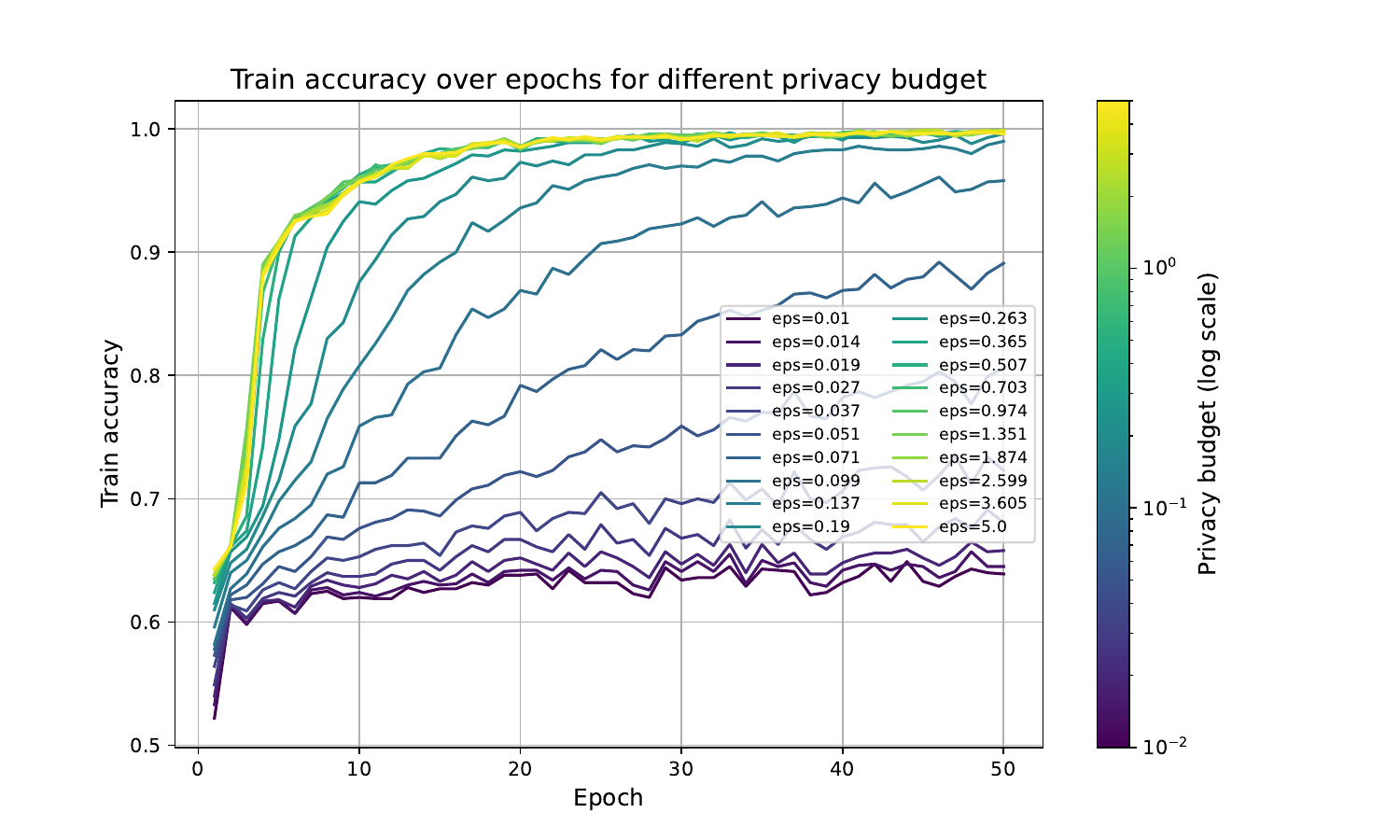}}%
    \subfloat[Test accuracy]{\includegraphics[width=\columnwidth]{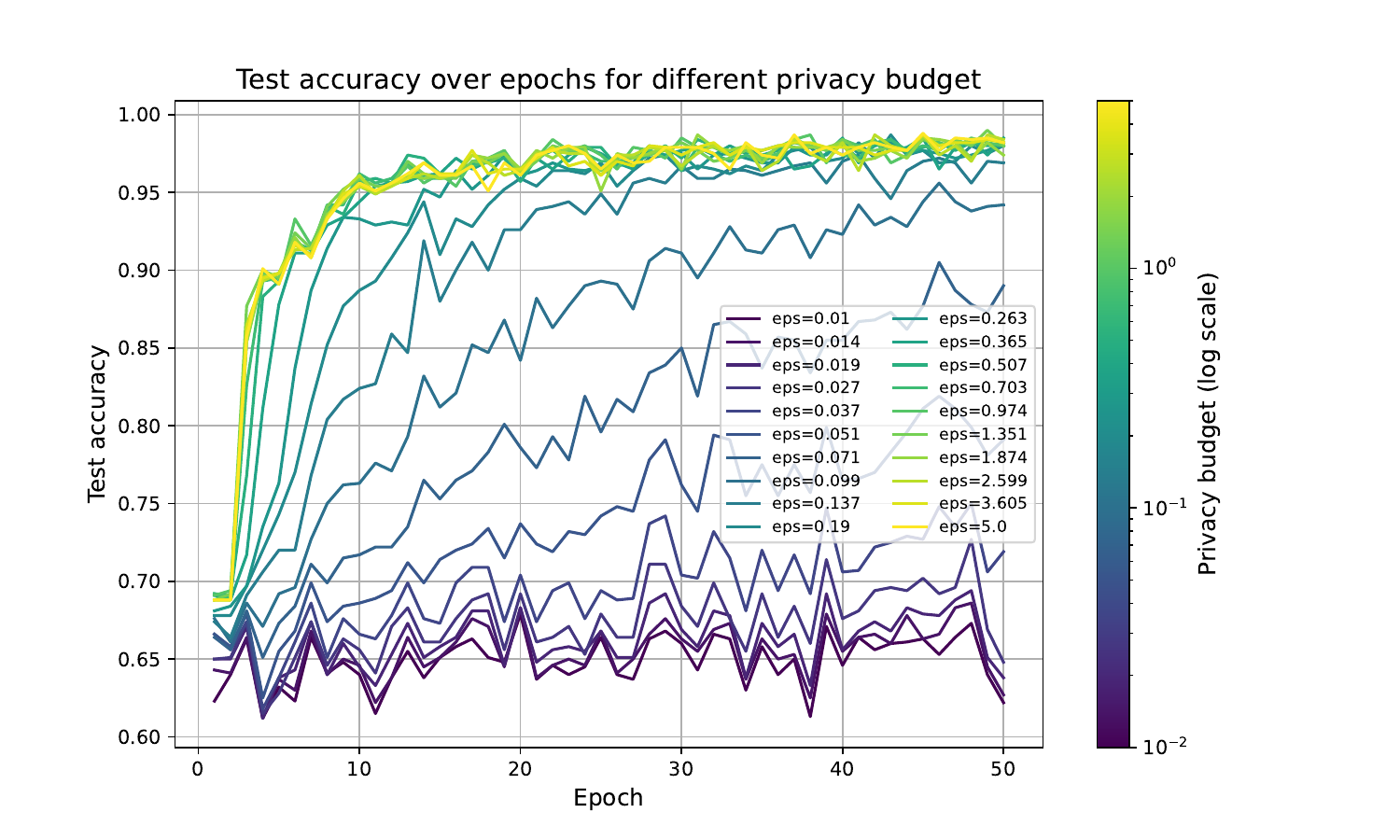}}%
    \caption{Accuracy in training and testing across epochs for different privacy budgets.}
    \label{best}
\end{figure*}


\subsection{Multimodal Feature Extraction}
\label{cp4: multimodal feature}
In our study, we use the “bert-base-uncased” model \cite{devlin2018bert} for processing EEG data and the “ViT-B/32” variant \cite{dosovitskiy2020image} of Vision Transformer (ViT) model to handle OM data. Throughout the experimental phase, we conduct a comparative analysis of two Bert model configurations \cite{devlin2018bert}: “bert-base-cased” and “bert-base-uncased”. Our findings indicate that the “bert-base-uncased” model excelled, enhancing the final classification performance for FoG by 1.1\%.

Furthermore, we explore various architectures for the image feature extraction model, including “resnet34” \cite{He2015}, and ViT variants of “ViT-B/16”, and “ViT-B/32” \cite{dosovitskiy2020image}. Through rigorous testing, we determine that “ViT-B/32” significantly outperform the others, delivering a 3.2\% improvement in classification accuracy over “resnet34” and a 0.4\% increase over “ViT-B/16”. 

As for cross-attention layers, we also test for single stream and double stream, where the single stream is constructed with 12 layers of transformer encoder of Bert, and the double stream is constructed with 3 layers of decoder with cross-attention. The latter outperforms 3.7\% than the previous one. 

\subsection{Randomness Allocation}
\label{cp4: randomness_exp}
In our research, we conduct an experiment of randomness allocation within the context of privacy enhancement on the application of dropout techniques and the integration of Laplacian noise. Our findings shown in \Cref{cp4: fig_weight} reveal distinct preferences in the application of these methods across different types of features. On average, EEG features exhibit an average 59.8\% dropout rate and an average 0.601 Laplacian noise scale, whereas Other Modal (OM) features show an average 62.2\% dropout rate and an average 0.583 Laplacian noise scale. In contrast, cross-modal (CM) features demonstrate a considerably lower preference for dropout with an average of 57.2\% rate and a higher preference for Laplacian noise injection with an average 0.620 noise scale.

In \Cref{cp4: fig_weight}, comparing the last column with the first column, the curve representing the dropout rate appears to be an inverted version of the curve for feature magnitude, corresponding to EEG, OM, and CM, while the curve representing the Laplacian noise scale appears to have a similar trend as the curve for feature magnitude by comparing the last column with the second one. This suggests a relationship between feature magnitude and the preference for randomness allocation via dual mechanism. Specifically, features with larger magnitudes are more likely to enhance privacy by incorporating noise. Conversely, features with smaller magnitudes are more prone to being dropped out directly.

Often larger feature magnitude suggests greater importance. Following this concept, the importance of features across the three modalities can be roughly ranked as CM $>$ EEG $>$ OM, underscoring the performance benefits derived from our cross-modal features. This ranking suggests that our approach leverages cross-modal features effectively, enhancing overall performance. 


\begin{figure*}[htbp]
\centerline{\includegraphics[width=2\columnwidth]{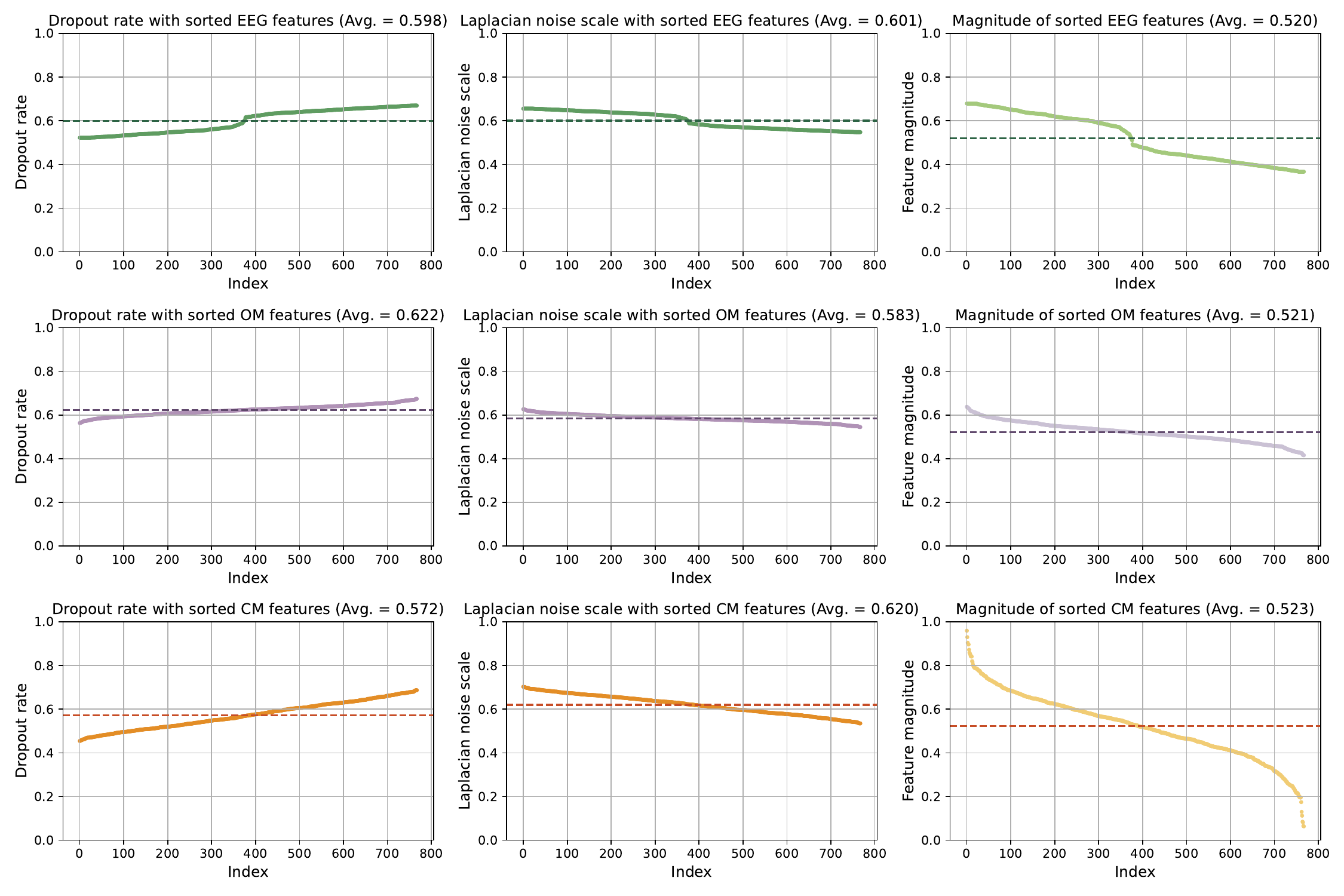}}
\caption{Dropout rate, Laplacian noise scale and feature magnitude of EEG features (first row), OM features (second row), and CM features (third row) features.}
\label{cp4: fig_weight}
\end{figure*}

\section{Conclusions and Discussion}
\label{Conclusions}
To conclude, we develop a novel Differentially Private Multimodal Laplacian Dropout (DP-MLD) for EEG representative learning. In the design of multimodal feature extraction, we treat EEG data similarly to text in language models and interpret physiological signals using techniques akin to Vision Transformers. We propose a cross-attention mechanism to extract cross-modal features. To ensure privacy at the feature level, we employ a dual-effect approach with 
dropout and Laplacian noise addition. This innovative strategy effectively balances privacy and analytical performance by strategically optimizing the allocation of randomness between dropout and noise injection. 

In experiment, We apply DP-MLD for disease diagnosis, specifically focusing on the detection of Freezing of Gait in Parkinson’s disease on an open-source FoG dataset \cite{zhang2022multimodal}. The experimental results demonstrate the outstanding performance of our method. DL-MLD shows an approximate 4\% improvement with 98.7\% accuracy with a common standard privacy budget $\varepsilon = 1.0$. DL-MLD can even still achieves on accuracy of 95.6\% with a stricter privacy budget with $\varepsilon = 0.1$. The advantage of our multimodal EEG representative learning scheme is also validated through a “non-private” variant of DP-MLD, which shows an excellent performance of 99.30\% accuracy. Our experiments also demonstrate a classic trade-off between performance and privacy. With larger privacy budgets, the performance is better with higher accuracy and convergence rate under the compromise of privacy guarantee, and vice versa. Besides, experiments demonstrate superior performance with Bert with initial weights of “bert-base-uncased” model \cite{devlin2018bert} for processing EEG data, Vision Transformer (ViT) model with initial weights “ViT-B/32” \cite{dosovitskiy2020image} to handle OM data, and 3-layer decoders for EEG, OM, and cross-modal feature extraction, respectively. Furthermore, in the experiments regarding randomness allocation of the dual scheme, we observe that features with larger magnitude have a stronger preference for Laplacian noise injection rather than dropout. Also the advantage of our extracted CM features is demonstrated. 

Looking ahead, we may generalize DP-MLD to other multimodal frameworks, additional tasks, and datasets. We may also generalize it with minor changes, such as other privacy schemes including moment accounting \cite{abadi2016deep} and local differential privacy \cite{asi2022optimal}. These future work may broaden the applicability of DP-MLD to other clinical data analysis and beyond, potentially extending into other fields.


\end{document}